\documentclass[letterpaper,british,twocolumn, nofootinbib]{revtex4}
\usepackage[T1]{fontenc}
\usepackage[latin9]{inputenc}
\usepackage{color}
\usepackage{babel}
\usepackage{textcomp}
\usepackage{amsmath}
\usepackage{amsthm}
\usepackage{amssymb}
\usepackage[unicode=true,pdfusetitle,
 bookmarks=true,bookmarksnumbered=false,bookmarksopen=false,
 breaklinks=true,pdfborder={0 0 0},backref=false,colorlinks=true]
 {hyperref}
\usepackage{breakurl}

\makeatletter

\@ifundefined{textcolor}{}
{%
 \definecolor{BLACK}{gray}{0}
 \definecolor{WHITE}{gray}{1}
 \definecolor{RED}{rgb}{1,0,0}
 \definecolor{GREEN}{rgb}{0,1,0}
 \definecolor{BLUE}{rgb}{0,0,1}
 \definecolor{CYAN}{cmyk}{1,0,0,0}
 \definecolor{MAGENTA}{cmyk}{0,1,0,0}
 \definecolor{YELLOW}{cmyk}{0,0,1,0}
}
\theoremstyle{plain}
\newtheorem{thm}{\protect\theoremname}

\usepackage{bbm}

\usepackage{amsmath}

\newcommand\eqdef{\mathrel{\overset{\makebox[0pt]{\mbox{\normalfont\tiny def}}}{=}}}

\usepackage{footmisc}

\makeatother

\providecommand{\theoremname}{Theorem}

\begin{document}

\title{Quantum Superpositions Cannot be Epistemic}

\author{John-Mark A. Allen}

\email{john-mark.allen@cs.ox.ac.uk}

\affiliation{Department of Computer Science, University of Oxford, Wolfson Building,
Parks Road, Oxford, OX1 3QD, United Kingdom. }

\date{May 2015}
\begin{abstract}
Quantum superposition states are behind many of the curious phenomena
exhibited by quantum systems, including Bell non-locality, quantum
interference, quantum computational speed-up, and the measurement
problem. At the same time, many qualitative properties of quantum
superpositions can also be observed in classical probability distributions
leading to a suspicion that superpositions may be explicable as probability
distributions over less problematic states; that is, a suspicion that
superpositions are \emph{epistemic}. Here, it is proved that, for
any quantum system of dimension $d>3$, this cannot be the case for
almost all superpositions. Equivalently, any underlying ontology must
contain ontic superposition states. A related question concerns the
more general possibility that some pairs of non-orthogonal quantum
states $|\psi\rangle,|\phi\rangle$ could be ontologically indistinct
(there are ontological states which fail to distinguish between these
quantum states). A similar method proves that if $|\langle\phi|\psi\rangle|^{2}\in(0,\frac{1}{4})$
then $|\psi\rangle,|\phi\rangle$ must approach ontological distinctness
as $d\rightarrow\infty$. The robustness of these results to small
experimental error is also discussed.
\end{abstract}
\maketitle

\section{Introduction\label{sec:Introduction}}

Is the quantum state \emph{ontic} (a state of reality) or \emph{epistemic}
(a state of knowledge)? This, rather old, question is the subject
of the now-famous PBR theorem \citep{Pusey12}, which proves that
the quantum state of a system is ontic given reasonable assumptions
about the ontic structure of multi-partite systems. Whilst these assumptions
appear weak and well-motivated, they have also been frequently challenged
and, as a result, many recent papers have sought to address the onticity
of the quantum state using only single-system arguments \citep{Maroney12a,Leifer13b,Barrett14,Leifer14a,Branciard14,Ballentine14}.
These theorems and discussions are reviewed in Ref. \citep{Leifer14b}.

All of this work addresses the \emph{epistemic realist}, who assumes
that a physical system is always in some definite \emph{ontic state}
(realist) and hopes that uncertainty about the ontic state might explain
certain features of quantum systems (epistemic). The features that
the epistemic realist might like to explain in this way include: indistinguishability
of non-orthogonal states, no-cloning, stochasticity of measurement
outcomes, and the exponential increase in state complexity with increasing
system size \citep{Spekkens07}. Preparing some quantum state $|\psi\rangle$
must result in some ontic state $\lambda$ obtaining, so some probability
distribution, called a \emph{preparation distribution}, must describe
the probabilities with which each $\lambda$ obtains in that preparation.
In general, preparation distributions for some pair of non-orthogonal
quantum states might overlap---there might be ontic states accessible
by preparing either of those quantum states. The main strategy of
the single-system ontology arguments is to prove that, in order to
preserve quantum predictions, these overlaps must be unreasonably
small---too small to explain any quantum features.

This paper initially concentrates on quantum \emph{superposition states}
defined with respect to some specified orthonormal basis (ONB). Superpositions
are behind quantum interference, the uncertainty principle, wave-particle
duality, entanglement, Bell non-locality \citep{Bell87}, and the
probable increased computational power of quantum theory \citep{Jozsa03}.
Perhaps most alarmingly, superpositions give rise to the measurement
problem, so captivatingly illustrated by the ``Schrödinger's cat''
thought experiment.

Schrödinger's cat is set up to be in a superposition of $|{\rm dead}\rangle$
and $|{\rm alive}\rangle$ quantum states. The epistemic realist (and
probably the cat) would ideally prefer the ontic state of the cat
to only ever be one of ``dead'' or ``alive'' (\emph{viz.,} only
in ontic states accessible to either the $|{\rm dead}\rangle$ or
$|{\rm alive}\rangle$ quantum states). In that case, the cat's apparent
quantum superposition would be \emph{epistemic}---there would be nothing
ontic about the superposition state. Conversely, if there are ontic
states which can only obtain when the cat is in a quantum superposition
(and never when the cat is in either quantum $|{\rm alive}\rangle$
or $|{\rm dead}\rangle$ states) then the superposition is unambiguously
\emph{ontic}: there are ontological features which correspond to that
superposition but not to non-superpositions, so that superposition
is real.

Obviously quantum superpositions are different from proper mixtures
of basis states. The question here is rather whether quantum superpositions
over basis states can be understood as probability distributions over
some subset of underlying ontic states, where each such ontic state
is also accessible by preparing some basis state.

The epistemic realist perspective on the foundations of quantum theory
is not only philosophically attractive but also appears to be tenable.
Theories in which the quantum state is explained in an epistemically
realist manner have been demonstrated to reproduce interesting subsets
of quantum theory which include characteristically quantum features
\citep{Spekkens14,Spekkens07,Bartlett12,Leifer15}. Moreover, they
include theories where superpositions are not ontic in the sense described
above. The question of the reality of superpositions in quantum theory
is, therefore, very much open.

For example, in Spekkens' toy theory \citep{Spekkens07} the ``toy-bit''
reproduces a subset of qubit behaviour. A toy bit consists of four
ontic states, \textbf{$a,b,c,d$}, and four possible preparations,
$|0),|1),|+),|-)$, which are analogous to the correspondingly named
qubit states. Each preparation corresponds to a uniform probabilistic
distribution over exactly two ontic states: $|0)$ is a distribution
over $a$ and $b$; $|1)$ a distribution over $c$ and $d$; $|+)$
over $a,c$; and $|-)$ over $b,d$. Full details of how these states
behave and how they reproduce qubit phenomena is described in Ref.~\citep{Spekkens07}.
For the purposes here, it suffices to notice that all ontic states
corresponding to the superpositions states $|+)$ and $|-)$ are also
ontic states corresponding to either $|0)$ or $|1)$---this toy theory
has nothing on the ontological level which can be identified as a
superposition so the superpositions are epistemic. Such models, therefore,
lend credibility to the idea that quantum superpositions themselves
might, in a similar way, fail to have an ontological basis.

Previous single-system theorems that bound ontic overlaps to argue
for the onticity of the quantum state \citep{Maroney12a,Leifer13b,Barrett14,Leifer14a,Branciard14}
share at least these shortcomings: (i) they prove that there exists
some pair of quantum states (taken from a specific set) with bounded
overlap, rather than bounding overlaps between arbitrary quantum states
and (ii) when the overlaps are proved to approach zero in some limit,
the quantum states involved also approach orthogonality in that same
limit \citep{Leifer14b}. 

In this paper it is proved that, for a $d>3$ dimensional quantum
system, almost all quantum superpositions with respect to any given
ONB must be ontic. A very similar argument can be used to obtain a
general bound on ontic overlaps for $d>3$, which addresses the above
shortcomings. Finally, the noise tolerance of these results is discussed.

\section{Ontological models}

The appropriate framework for discussing epistemic realism is that
of ontological models \citep{Harrigan10,Harrigan07,Leifer14b}. It
is flexible enough for most realist approaches to quantum ontology
to be cast as ontological models \citep{Ballentine14} including,
but not limited to, Bohmian theories, spontaneous collapse theories,
and naïve wave-function-realist theories.\footnote{Conversely, ontological models are irrelevant for any ``anti-realist'',
``instrumentalist'', ``positivist'', or ``Copenhagen-like''
theories denying the existence of an underlying ontology. For example,
quantum-Bayesian theories are exempt from ontological model analysis.}

An ontological model of a system has a set $\Lambda$ of \emph{ontic
states $\lambda\in\Lambda$}. The ontic state which the system occupies
dictates the properties and behaviour of the system, regardless of
any other theory (such as quantum theory) which may be used to describe
it.

An ontological model for a quantum system is constrained by the fact
that it must reproduce the predictions of quantum theory (at least
where they are empirically verifiable). Recall that a quantum system
is described with a $d$-dimensional complex Hilbert space $\mathcal{H}$
with $\mathcal{P}(\mathcal{H})\eqdef\{|\psi\rangle\in\mathcal{H}\,:\,\left\Vert \psi\right\Vert =1,|\psi\rangle\sim\mathrm{e}^{i\theta}|\psi\rangle\}$
as the set of distinct pure quantum states.\footnote{For simplicity, take $d<\infty$.}
Quantum superpositions are defined with respect to some ONB $\mathcal{B}$
of $\mathcal{H}$ and are simply those $|\psi\rangle\in\mathcal{P}(\mathcal{H})$
for which $|\psi\rangle\not\in\mathcal{B}$.

The preparation distributions\footnote{In fact, this treatment of ontological models is not as general as
it should be. Reference~\citep{Leifer14b} notes that, instead of
probability \emph{distributions}, one should consider general probability
\emph{measures} $\mu$ over a measurable space $(\Lambda,\Sigma)$
and ontological models can be re-formulated measure-theoretically.
The presentation here implicitly, and problematically, assumes some
canonical measure $\mathrm{d}\lambda$ over $\Lambda$ with respect
to which all of the probability distributions can be defined. It is
possible to derive the results presented here in the more rigorous
formulation, but doing so would be at the expense of conceptual clarity.
In light of this simplification some of the proofs presented here
will also lack in mathematical rigour at certain steps, though more
thorough versions of the same results can be derived.} $\mu(\lambda)$ for some state $|\psi\rangle\in\mathcal{P}(\mathcal{H})$
form a set $\Delta_{|\psi\rangle}$ since different ways of preparing
the same $|\psi\rangle$ may result in different distributions $\mu\in\Delta_{|\psi\rangle}$.
If $\Delta_{|\psi\rangle}$ is a singleton for every $|\psi\rangle\in\mathcal{P}(\mathcal{H})$,
then the ontological model is \emph{preparation non-contextual}\footnote{Preparation non-contextuality for pure states is often implicitly
assumed because it rarely affects arguments \citep{Leifer14b}. Rather,
preparation contextuality for \emph{mixed }quantum states is more
often discussed \citep{Spekkens05}. However, explicit preparation
contextuality for pure states will be needed here.}\emph{ }for pure states (otherwise, it is \emph{preparation contextual}).
Let $\Lambda_{\mu}\eqdef\{\lambda\in\Lambda\,:\,\mu(\lambda)>0\}$
be the support of the distribution $\mu$.

A measurement $M$ of a quantum system can be represented as a set
of outcomes: either vectors of some ONB $\mathcal{B}^{\prime}$ (for
an ONB measurement) or POVM elements (for a general POVM measurement).
An ontological model assigns a set $\Xi_{M}$ of conditional probability
distributions, called \emph{response functions} $\mathbb{P}_{M}\in\Xi_{M}$,
to $M$. A method for performing measurement $M$ selects some $\mathbb{P}_{M}\in\Xi_{M}$
which gives the probability of obtaining outcome $E\in M$ conditional
on the ontic state of the system. A preparation of $|\psi\rangle$
via $\mu\in\Delta_{|\psi\rangle}$ followed by a measurement $M$
via $\mathbb{P}_{M}\in\Xi_{M}$, therefore, returns outcome $E\in M$
with probability 
\begin{equation}
\mathbb{P}_{M}(E\,|\,\mu)=\int_{\Lambda}\mathrm{d}\lambda\,\mu(\lambda)\mathbb{P}_{M}(E\,|\,\lambda).
\end{equation}

Transformations acting on a system must correspond to stochastic maps
on its space of ontic states $\Lambda$. An ontological model assigns
a set $\Gamma_{U}$ of stochastic maps $\gamma$ to each unitary transformation
$U$ over $\mathcal{H}$. A method for performing $U$ selects some
$\gamma\in\Gamma_{U}$ which, given that the system is in ontic state
$\lambda^{\prime}$, describes a probability distribution $\gamma(\cdot|\lambda^{\prime})$,
so that the probably that $\lambda^{\prime}$ is mapped to $\lambda$
under this operation is $\gamma(\lambda|\lambda^{\prime})$. A preparation
of $|\psi\rangle$ via $\mu\in\Delta_{|\psi\rangle}$ followed by
a transformation $U$ via $\gamma\in\Gamma_{U}$ results in an ontic
state distributed according to the distribution $\nu$, given by
\begin{equation}
\nu(\lambda)=\int_{\Lambda}\mathrm{d}\lambda^{\prime}\mu(\lambda^{\prime})\gamma(\lambda\,|\,\lambda^{\prime}),\quad\forall\lambda\in\Lambda.\label{eq:ontological-transformation}
\end{equation}
It is required that $\nu\in\Delta_{U|\psi\rangle}$, since this an
example of a procedure preparing the quantum state $U|\psi\rangle$.

For now, assume that measurement statistics predicted by quantum theory
are exactly correct, so valid ontological models for quantum systems
must reproduce them. Therefore, for every $|\psi\rangle\in\mathcal{P}(\mathcal{H})$,
every unitary $U$ over $\mathcal{H}$, and every measurement $M$,
any choices of preparation $\mu\in\Delta_{|\psi\rangle}$, stochastic
map $\gamma\in\Gamma_{U}$, and response function $\mathbb{P}_{M}\in\Xi_{M}$,
must satisfy 
\begin{equation}
\langle\psi|U^{\dagger}EU|\psi\rangle=\int_{\Lambda}\mathrm{d}\lambda\int_{\Lambda}\mathrm{d}\lambda^{\prime}\mu(\lambda^{\prime})\gamma(\lambda|\lambda^{\prime})\mathbb{P}_{M}(E\,|\,\lambda),\quad\forall E\in M.\label{eq:exact-quantum-statistics}
\end{equation}

It shall be useful to consider the \emph{stabiliser subgroups} of
unitaries $\mathcal{S}_{|\psi\rangle}\eqdef\{U\,:\,U|\psi\rangle=|\psi\rangle\}$
for each $|\psi\rangle\in\mathcal{P}(\mathcal{H})$. In particular,
an ontological model is \emph{preparation non-contextual with respect
to stabiliser unitaries of $|\psi\rangle$} if and only if for every
$\mu\in\Delta_{|\psi\rangle}$, $U\in\mathcal{S}_{|\psi\rangle}$,
and $\gamma\in\Gamma_{U}$ the action of $\gamma$, according to Eq.~(\ref{eq:ontological-transformation}),
leaves the preparation distribution $\mu$ unaffected (that is, $\nu$
in Eq.~(\ref{eq:ontological-transformation}) would be equal to $\mu$).

\section{Measuring overlaps}

One way to quantify the overlap between preparation distributions
is the \emph{asymmetric overlap} $\varpi(|\phi\rangle|\mu)$ \citep{Maroney12a,Leifer13b,Ballentine14},
defined as the probability of obtaining an ontic state $\lambda$
accessible from some preparation of $|\phi\rangle$ when sampling
from $\mu$. Formally, 
\begin{equation}
\varpi(|\phi\rangle\,|\,\mu)\eqdef\int_{\Lambda_{|\phi\rangle}}\mathrm{d}\lambda\,\mu(\lambda)\label{eq:asymmetric-overlap}
\end{equation}
where $\Lambda_{|\phi\rangle}\eqdef\cup_{\nu\in\Delta_{|\phi\rangle}}\Lambda_{\nu}$
is the total support of all possible preparations of $|\phi\rangle$.
By Eq.~(\ref{eq:exact-quantum-statistics}), the asymmetric overlap
must be upper bounded by the Born rule measurement probability (proof
in appendix~\ref{sec:Proofs})
\begin{equation}
\varpi(|\phi\rangle\,|\,\mu)\leq|\langle\phi|\psi\rangle|^{2},\quad\forall\mu\in\Delta_{|\psi\rangle}.\label{eq:asymmetric-trivial-bound}
\end{equation}
That is, the probability of obtaining a $\lambda$ compatible with
$|\phi\rangle$ when preparing $|\psi\rangle$ cannot exceed the probability
of getting the measurement outcome $|\phi\rangle$ having prepared
$|\psi\rangle$.

This quantifies overlaps between pairs of quantum states, but what
of multi-partite overlaps? The asymmetric multi-partite overlap $\varpi(|0\rangle,|\phi\rangle,...\,|\,\mu)$
acts like the union of the bipartite overlaps $\varpi(|0\rangle|\mu)$,
$\varpi(|\phi\rangle|\mu)$, etc. It is defined as the probability
of obtaining a $\lambda\in\Lambda_{|0\rangle}\cup\Lambda_{|\phi\rangle}\cup...$
when sampling from $\mu$. Formally,
\begin{equation}
\varpi(|0\rangle,|\phi\rangle,...\,|\,\mu)\eqdef\int_{\Lambda_{|0\rangle}\cup\Lambda_{|\phi\rangle}\cup...}\mathrm{d}\lambda\,\mu(\lambda).\label{eq:multipartite-asymmetric}
\end{equation}
From Eqs.~(\ref{eq:asymmetric-overlap},\ref{eq:multipartite-asymmetric})
and Boole's inequality, it is clear that 
\begin{equation}
\varpi(|0\rangle,|\phi\rangle,...\,|\,\mu)\leq\varpi(|0\rangle\,|\,\mu)+\varpi(|\phi\rangle\,|\,\mu)+...\label{eq:multipartite-anti-distinguishable-bound}
\end{equation}

Quantum states are only perfectly distinguishable if they are mutually
orthogonal. Distinguishable states must be ontologically distinct
(their preparation distributions cannot overlap) in order to satisfy
Eq.~(\ref{eq:exact-quantum-statistics}). The opposite concept of
\emph{anti-distinguishability} is much more useful in discussions
of ontic overlaps \citep{Leifer14b}. A set $\{|\psi\rangle,|\phi\rangle,...\}\subset\mathcal{P}(\mathcal{H})$
is anti-distinguishable if and only if there exists a measurement
$M=\{E_{\text{\textlnot}\psi},E_{\text{\textlnot}\phi},...\}$ such
that 
\begin{equation}
\langle\psi|E_{\text{\textlnot}\psi}|\psi\rangle=\langle\phi|E_{\text{\textlnot}\phi}|\phi\rangle=...=0,\label{eq:anti-distinguishing-measurement}
\end{equation}
\emph{i.e.} the measurement can tell, with certainty, one state from
the set that was \emph{not} prepared. It has been proven \citep{Caves02,Barrett14}
that if some inner products $a=|\langle\phi|\psi\rangle|^{2}$, $b=|\langle0|\psi\rangle|^{2}$,
$c=|\langle0|\phi\rangle|^{2}$ satisfy
\begin{equation}
a+b+c<1,\quad(1-a-b-c)^{2}\geq4abc,\label{eq:antidistinguishing}
\end{equation}
then the triple $\{|\psi\rangle,|\phi\rangle,|0\rangle\}$ must be
anti-distinguishable. Anti-distinguishable triples $\{|\psi\rangle,|\phi\rangle,|0\rangle\}$
are useful because $\Lambda_{|\psi\rangle}\cap\Lambda_{|\phi\rangle}\cap\Lambda_{|0\rangle}=\emptyset$
and therefore $\varpi(|0\rangle,|\phi\rangle|\mu)=\varpi(|0\rangle|\mu)+\varpi(|\phi\rangle|\mu)$
for all $\mu\in\Delta_{|\psi\rangle}$, as proved in appendix~\ref{sec:Proofs}.

\section{Quantum superpositions are real}

Define quantum superpositions with respect to some ONB $\mathcal{B}$
and consider any superposition state $|\psi\rangle\not\in\mathcal{B}$.
If every ontic state accessible by preparing any $\mu\in\Delta_{|\psi\rangle}$
is also accessible by preparing some $|i\rangle\in\mathcal{B}$, then
$|\psi\rangle$ has no ontology independent of $\mathcal{B}$ in the
ontological model. Such a $|\psi\rangle$ is called an \emph{epistemic
}or \emph{statistical superposition }and must satisfy
\begin{eqnarray}
\sum_{|i\rangle\in\mathcal{B}}\varpi(|i\rangle\,|\,\mu) & = & 1,\;\forall\mu\in\Delta_{|\psi\rangle},\;\text{or equivalently,}\\
\varpi(|i\rangle\,|\,\mu) & = & |\langle i|\psi\rangle|^{2},\quad\forall|i\rangle\in\mathcal{B},\mu\in\Delta_{|\psi\rangle}.\label{eq:epistemic-superposition}
\end{eqnarray}

The alternative occurs when there exists some subset of ontic states
$\lambda\in\Lambda_{\psi}^{\mathcal{B}}\subset\Lambda$ for which
$\mu(\lambda)>0$ for some $\mu\in\Delta_{|\psi\rangle}$, but $\nu(\lambda)=0$
for every $\nu\in\Delta_{|i\rangle\in\mathcal{B}}$. That is, the
ontic states in $\Lambda_{\psi}^{\mathcal{B}}$ are accessible by
preparing $|\psi\rangle$ but not by preparing any $|i\rangle\in\mathcal{B}$,
making $|\psi\rangle$ an \emph{ontic} or \emph{real superposition}.

From Eqs.~(\ref{eq:asymmetric-trivial-bound},\ref{eq:epistemic-superposition}),
a superposition $|\psi\rangle\not\in\mathcal{B}$ can only be epistemic
if the asymmetric overlap $\varpi(|i\rangle\,|\,\mu)$ is maximal
for every $\mu\in\Delta_{|\psi\rangle}$ and all $|i\rangle\in\mathcal{B}$.
Therefore, the statement that ``not every quantum superposition can
be epistemic'' is rather weak. A more interesting question is whether
an individual superposition state $|\psi\rangle\in\mathcal{B}$ can
be epistemic.
\begin{thm}
\label{thm:superpositions-are-real}Consider a quantum system of dimension
$d>3$ and define superpositions with respect to some ONB $\mathcal{B}$.
Almost all quantum superposition states $|\psi\rangle\not\in\mathcal{B}$
are ontic.\end{thm}
\begin{proof}
Let $|\psi\rangle$ be an arbitrary superposition state $|\psi\rangle\not\in\mathcal{B}$
and assume only that $|\psi\rangle$ is not an exact 50:50 superposition
of two states in $\mathcal{B}$. This is true for almost all superpositions
and guarantees that there exists some $|0\rangle\in\mathcal{B}$ such
that $|\langle0|\psi\rangle|^{2}\in(0,\frac{1}{2})$.

Define an ONB $\mathcal{B}^{\prime}=\{|0\rangle\}\cup\{|i^{\prime}\rangle\}_{i=1}^{d-1}$
containing this $|0\rangle$ such that 
\begin{equation}
|\psi\rangle=\alpha|0\rangle+\beta|1^{\prime}\rangle+\gamma|2^{\prime}\rangle\label{eq:Psi-construction}
\end{equation}
where $\alpha\in\mathbb{R}$, $\alpha\in(0,\,1/\sqrt{2})$, and $\beta\eqdef\sqrt{2}\alpha^{2}$.
Such bases always exists since $|\langle0|\psi\rangle|^{2}=\alpha^{2}$
and $|\alpha|^{2}+|\beta|^{2}=\alpha^{2}(1+2\alpha^{2})<1$. With
respect to the same $\mathcal{B}^{\prime}$, define
\begin{equation}
|\phi\rangle\eqdef\delta|0\rangle+\eta|1^{\prime}\rangle+\kappa|3^{\prime}\rangle\label{eq:Phi-construction}
\end{equation}
where $\delta\eqdef1-2\alpha^{2}$, $\eta\eqdef\sqrt{2}\alpha$. This
can always be normalised because $|\delta|^{2}+|\eta|^{2}=(1-2\alpha^{2})^{2}+2\alpha^{2}<1$.

The above construction has been chosen such that:\end{proof}
\begin{itemize}
\item $|\langle0|\psi\rangle|^{2}=\alpha^{2}=|\langle\phi|\psi\rangle|^{2}$
so there exists a unitary $U\in\mathcal{S}_{|\psi\rangle}$ for which
$U|0\rangle=|\phi\rangle$; and
\item the inner products $|\langle0|\psi\rangle|^{2}$, $|\langle\phi|\psi\rangle|^{2}$,
$|\langle0|\phi\rangle|^{2}$ satisfy Eq.~(\ref{eq:antidistinguishing})
and therefore the triple $\{|\psi\rangle,|\phi\rangle,|0\rangle\}$
is anti-distinguishable.\end{itemize}
\begin{proof}
For any preparation distribution $\mu^{\prime}\in\Delta_{|\psi\rangle}$
of $|\psi\rangle$, consider $\varpi(|0\rangle|\mu^{\prime})$. For
any unitary $V$ and any corresponding $\gamma\in\Gamma_{V}$, $\mu^{\prime}$
is evolved to some $\mu\in\Delta_{V|\psi\rangle}$ as in Eq.~(\ref{eq:ontological-transformation}).
This operation cannot decrease the asymmetric overlap $\varpi(V|0\rangle|\mu)\geq\varpi(|0\rangle|\mu^{\prime})$
and, in particular, letting $V=U$ one finds 
\begin{equation}
\varpi(|\phi\rangle|\mu)\geq\varpi(|0\rangle|\mu^{\prime}).\label{eq:first-thm-monotonicity}
\end{equation}
A proof of this is provided in appendix~\ref{sec:Proofs}. Therefore,
there must exist preparation distributions $\mu,\mu^{\prime}\in\Delta_{|\psi\rangle}$
satisfying Eq.~(\ref{eq:first-thm-monotonicity}).

Assume towards a contradiction that $|\psi\rangle$ is an epistemic
superposition so that Eq.~(\ref{eq:epistemic-superposition}) holds
and, in particular, $\varpi(|0\rangle|\mu)=\varpi(|0\rangle|\mu^{\prime})=\alpha^{2}$.
By Eq.~(\ref{eq:first-thm-monotonicity}) it is therefore found that
\begin{equation}
\varpi(|\phi\rangle\,|\,\mu)\geq\varpi(|0\rangle\,|\,\mu).\label{eq:first-thm-phi-overlap}
\end{equation}

Consider, then, a preparation of the state $|\psi\rangle$ via $\mu$
followed by an ONB measurement $M$ in the $\mathcal{B}^{\prime}$
basis. Since $|\psi\rangle$ was prepared, $\lambda\in\Lambda_{|\psi\rangle}$
and the only possible measurement outcomes are $|0\rangle$, $|1^{\prime}\rangle$,
and $|2^{\prime}\rangle$. By Eq.~(\ref{eq:exact-quantum-statistics}),
almost all $\lambda\in\Lambda_{|0\rangle}$ must return the outcome
$|0\rangle$ with certainty. Similarly, almost all $\lambda\in\Lambda_{|\phi\rangle}$
can only return $|0\rangle$, $|1^{\prime}\rangle$, or $|3^{\prime}\rangle$
as the measurement outcome. Therefore, the probability of obtaining
outcomes $|0\rangle$ or $|1^{\prime}\rangle$ must be lower bounded
by the probability of obtaining a $\lambda\in\Lambda_{|0\rangle}\cup\Lambda_{|\phi\rangle}$;
formally,
\begin{eqnarray}
\mathbb{P}_{M}(|0\rangle\vee|1^{\prime}\rangle\,|\,\mu) & \geq & \varpi(|0\rangle,|\phi\rangle\,|\,\mu)=\varpi(|0\rangle\,|\,\mu)+\varpi(|\phi\rangle\,|\,\mu)\nonumber \\
 & \geq & 2\varpi(|0\rangle\,|\,\mu)\label{eq:thm-asymmetric-bound}
\end{eqnarray}
where the equality follows because $\{|0\rangle,|\psi\rangle,|\phi\rangle\}$
is anti-distinguishable and the final line follows from Eq.~(\ref{eq:first-thm-phi-overlap}),
which is found by assuming that $|\psi\rangle$ is an epistemic superposition.

In order to satisfy Eq.~(\ref{eq:exact-quantum-statistics}) 
\begin{equation}
\mathbb{P}_{M}(|0\rangle\vee|1^{\prime}\rangle\,|\,\mu)=|\langle0|\psi\rangle|^{2}+|\langle1^{\prime}|\psi\rangle|^{2}=\alpha^{2}+2\alpha^{4}.\label{eq:thm-probabilities}
\end{equation}
Combining Eqs.~(\ref{eq:thm-asymmetric-bound},\ref{eq:thm-probabilities})
it is found that
\begin{equation}
\varpi(|0\rangle\,|\,\mu)\leq\alpha^{2}\left(\frac{1}{2}+\alpha^{2}\right)<\alpha^{2}.\label{eq:thm1-conc}
\end{equation}
But, this contradicts the assumption that $|\psi\rangle$ is an epistemic
superposition which implies $\varpi(|0\rangle|\mu)=\alpha^{2}$ by
Eq.~(\ref{eq:epistemic-superposition}). Therefore, if the predictions
of quantum theory are to be exactly reproduced, any such $|\psi\rangle$
must be an ontic, rather than epistemic, superposition.
\end{proof}

\section{Bounds on general overlaps}

Theorem \ref{thm:superpositions-are-real} establishes the reality
of almost all superpositions in $d>3$ by bounding an asymmetric overlap.
This suggests that a similar method may be used to prove a general
bound on ontic overlaps.

Recall shortcomings (i) and (ii) of the previous single-system ontology
arguments as mentioned in Sec.~\ref{sec:Introduction}. Shortcoming
(i) leaves open the possibility that many pairs of quantum states
could have significant ontic overlaps, while (ii) casts doubt on the
significance of those zero-overlap limits (as orthogonal states are
distinguishable and therefore must be trivially ontologically distinct).

The following theorem address these shortcomings.
\begin{thm}
\label{thm:large-d-limit-theorem}Consider a $d>3$ dimensional quantum
system and any pair $|\psi\rangle,|0\rangle\in\mathcal{P}(\mathcal{H})$
such that $|\langle0|\psi\rangle|^{2}\eqdef\alpha^{2}\in(0,\frac{1}{4})$.
Assume that pure state preparations of $|\psi\rangle$ are non-contextual
with respect to stabiliser unitaries of $|\psi\rangle$. For any preparation
distribution $\mu\in\Delta_{|\psi\rangle}$, the asymmetric overlap
must satisfy
\begin{eqnarray}
\varpi(|0\rangle\,|\,\mu) & \leq & \alpha^{2}\left(\frac{1+2\alpha}{d-2}\right)\\
\lim_{d\rightarrow\infty}\varpi(|0\rangle\,|\,\mu) & = & 0
\end{eqnarray}
and so becomes arbitrarily small as $d$ increases, independently
of $\alpha$.
\end{thm}
The proof, in appendix~\ref{sec:Proofs}, closely follows that of
Thm.~\ref{thm:superpositions-are-real}. The assumption of pure state
preparation non-contextuality with respect to stabiliser unitaries
is required to replace the assumption used in Thm.~\ref{thm:superpositions-are-real}
that $|\psi\rangle$ is an epistemic superposition with respect to
$|0\rangle$.

\section{Noise tolerance}

Thus far Eq.~(\ref{eq:exact-quantum-statistics}) has been assumed,
demanding that quantum statistics are exactly reproduced by valid
ontological models. However, it is impossible to verify this. At most,
experiments demonstrate quantum probabilities hold to within some
finite additive error $\epsilon\in(0,1]$. It is therefore necessary
to consider \emph{noise-tolerant }versions of the above theorems.

Unfortunately, the asymmetric overlap is a noise intolerant quantity---there
exist simple ontological models in which every pair of quantum states
have unit asymmetric overlap and still reproduce quantum probabilities
to within any given $\epsilon\in(0,1]$. However, an alternative overlap
measure, the \emph{symmetric overlap} $\omega(|\psi\rangle,|\phi\rangle)$
\citep{Maroney12a,Barrett14,Leifer14a,Branciard14,Leifer14b}, is
robust to small errors and Thm.~\ref{thm:large-d-limit-theorem}
can be modified to bound the symmetric overlap in a noise-tolerant
way.

Suppose you are given some $\lambda\in\Lambda$ obtained by sampling
from either $\mu$ or $\nu$ (each with equal \emph{a priori }probability).
If you try to guess which of $\mu,\nu$ was used, then $\omega(\mu,\nu)/2$
is defined to be the average probability of error when using the optimal
strategy. This is known to correspond to \citep{Maroney12a,Barrett14}
\begin{equation}
\omega(\mu,\nu)\eqdef\int_{\Lambda}\mathrm{d}\lambda\min\{\mu(\lambda),\nu(\lambda)\}.
\end{equation}
Extending this to quantum states themselves, rather than to preparation
distributions, gives the symmetric overlap 
\begin{equation}
\omega(|\psi\rangle,|\phi\rangle)\eqdef\sup_{\mu\in\Delta_{|\psi\rangle},\nu\in\Delta_{|\phi\rangle}}\omega(\mu,\nu).\label{eq:quantum-state-symmetric}
\end{equation}

Quantum theory provides an upper bound on the symmetric overlap, since
any quantum procedure for distinguishing $|\psi\rangle,|\phi\rangle$
is also a method for distinguishing $\mu\in\Delta_{|\psi\rangle},\nu\in\Delta_{|\phi\rangle}$
in an ontological model. As $\frac{1}{2}\left(1-\sqrt{1-|\langle\phi|\psi\rangle|^{2}}\right)$
is the minimum average error probability when distinguishing $|\psi\rangle,|\phi\rangle$
within quantum theory\footnote{By using the Helstrom measurement \citep{Waldherr12,Barrett14}.}
it follows that $\omega(\mu,\nu)\leq1-\sqrt{1-|\langle\phi|\psi\rangle|^{2}}$
holds for every $\mu\in\Delta_{|\psi\rangle},\nu\in\Delta_{|\phi\rangle}$
and so 
\begin{equation}
\omega(|\psi\rangle,|\phi\rangle)\leq1-\sqrt{1-|\langle\phi|\psi\rangle|^{2}}.\label{eq:symmetric-state-trivial-bound}
\end{equation}

\begin{thm}
\label{thm:noisy-theorem}Consider the assumptions of Thm.~\ref{thm:large-d-limit-theorem},
but only assume that the probabilities predicted by quantum theory
are accurate to within $\pm\epsilon$, for some $\epsilon\in(0,1]$.
The symmetric overlap must satisfy
\begin{eqnarray}
\omega(|0\rangle,|\psi\rangle) & \leq & \alpha^{2}\left(\frac{1+2\alpha}{d-2}\right)+\frac{(3d^{2}-7d)}{2(d-2)}\epsilon.
\end{eqnarray}
This bound is tighter than Eq.~(\ref{eq:symmetric-state-trivial-bound})
for $d>5$ for small $\epsilon$.
\end{thm}
The proof is provided in appendix~\ref{sec:Proofs}. This theorem
makes Thm.~\ref{thm:large-d-limit-theorem} noise-tolerant at the
expense of weakening the bound (and only applying for $d>5$). This
is because the simple bound on symmetric overlap {[}Eq.~(\ref{eq:symmetric-state-trivial-bound}){]}
is lower than that for the asymmetric overlap {[}Eq.~(\ref{eq:asymmetric-trivial-bound}){]}
and therefore more difficult to improve upon.

Note that this theorem does not immediately imply that almost all
superpositions are real. However, by demonstrating that Thm\@.~\ref{thm:large-d-limit-theorem}'s
arguments can be made robust against error, it suggests that a noise-tolerant
version of Thm\@.~\ref{thm:superpositions-are-real} should also
be possible. Even so, a noise-tolerant version of Thm.~\ref{thm:superpositions-are-real}
would require the definition of ``epistemic superposition'' to be
modified, since it is currently defined in terms of the noise intolerant
asymmetric overlap and is therefore noise intolerant.

\section{Discussion}

Assuming that quantum statistics are exactly correct, Thm.~\ref{thm:superpositions-are-real}
proves that, for $d>3$, almost all superpositions defined with respect
to any given basis $\mathcal{B}$ must be real. Therefore, any epistemic
realist account of quantum theory must include ontic features corresponding
to superposition states. The unfortunate cat cannot be put out of
its misery.

A similar method and construction is used in Thm.~\ref{thm:large-d-limit-theorem}
to prove that, for arbitrary states satisfying $|\langle\phi|\psi\rangle|^{2}\in(0,\frac{1}{4})$,
ontic overlap must approach zero as $d$ increases for fixed $|\langle\phi|\psi\rangle|^{2}$.
Theorem~\ref{thm:noisy-theorem} makes this robust against small
errors in quantum probabilities, at the expense of weakening the bound.
Both theorems require an extra assumption: pure state preparation
non-contextuality with respect to stabiliser unitaries. Pure state
preparation contextuality is often implicitly assumed wholesale, so
this assumption should not be very controversial. Moreover, appendix~\ref{sec:Justifying-Preparation-Non-conte}
provides a heuristic argument to the effect that this type of contextuality
is a natural assumption in practice.

These results are damaging to any epistemic approach to quantum theory
compatible with the ontological models formalism that reproduces quantum
statistics exactly. Such a programme can never hope to epistemically
explain superpositions, including macroscopic superpositions. Moreover,
for any moderately large system, a large number of pairs of non-orthogonal
states cannot overlap significantly, making it unlikely that such
overlaps can satisfactorily explain quantum features.

As a result tolerant to small errors, it is possible that Thm.~\ref{thm:noisy-theorem}
could be experimentally tested. Such a test would require demonstration
of small errors in probabilities for a wide range of measurements
on a $d>5$ dimensional system.

The methodology of Thms.~\ref{thm:superpositions-are-real},\ref{thm:large-d-limit-theorem}
is tightly linked to the asymmetric overlap as a probability, making
noise-tolerant versions a challenge to extract. If the conclusion
from Thms.~\ref{thm:superpositions-are-real},\ref{thm:large-d-limit-theorem}
could be obtained though an operational methodology (closer to that
of Bell's theorem \citep{Bell87} or the PBR theorem \citep{Pusey12})
this would likely lead to better noise-tolerant extensions and better
opportunities for experimental investigation. Such an operational
version may also make it easier to discover any information theoretic
implications of these results.
\begin{acknowledgments}
I would like to thank Jonathan Barrett, Owen Maroney, Dominic C. Horsman,
and Matty Hoban for insightful discussions as well as an anonymous
referee for thorough and insightful comments. This work is supported
by: the Engineering and Physical Sciences Research Council (EPSRC);
the European Coordinated Research on Long-term Challenges in Information
and Communication Sciences \& Technologies (CHIST-ERA) project on
Device Independent Quantum Information Processing (DIQIP); and the
Foundational Questions Institute (FQXi) Large Grant ``Thermodynamic
vs information theoretic entropies in probabilistic theories''.
\end{acknowledgments}

\bibliographystyle{apsrev4-1}
\bibliography{references}

\onecolumngrid

\appendix

\section{Proofs\label{sec:Proofs}}

\subsection{Simple upper bound on asymmetric overlap}

To prove Eq.~(\ref{eq:asymmetric-trivial-bound}) from Eq.~(\ref{eq:exact-quantum-statistics}),
first consider any $|\phi\rangle\in\mathcal{P}(\mathcal{H})$ and
any ONB measurement $M\ni|\phi\rangle$ and no evolution between preparation
and measurement ($U=\mathbbm{1}$ and $\gamma$ is trivial). Equation~(\ref{eq:exact-quantum-statistics})
then gives
\begin{eqnarray}
1 & = & \int_{\Lambda}\mathrm{d}\lambda\,\nu(\lambda)\mathbb{P}_{M}(|\phi\rangle\,|\,\lambda)\\
 & = & \int_{\Lambda_{\nu}}\mathrm{d}\lambda\,\nu(\lambda)\mathbb{P}_{M}(|\phi\rangle\,|\,\lambda)
\end{eqnarray}
for any $\nu\in\Delta_{|\phi\rangle}$. This can only be the case
if $\mathbb{P}_{M}(|\phi\rangle\,|\,\lambda)=1$ for almost all $\lambda\in\Lambda_{\nu}$
and therefore\footnote{A more mathematically rigorous treatment would fully consider this
step in the light of $\Delta_{|\phi\rangle}$ being uncountable in
the general case. Such a discussion is omitted for the sake of conceptual
clarity and since a more rigorous treatment would also have to account
for the issues raised in footnote 3.} (since $\nu\in\Delta_{|\phi\rangle}$ is arbitrary) $\mathbb{P}_{M}(|\phi\rangle\,|\,\lambda)=1$
for almost all $\lambda\in\Lambda_{|\phi\rangle}$. In other words,
almost all ontic states in the support of any preparation $\nu$ of
$|\phi\rangle$ must return the measurement result $|\phi\rangle$
with certainty in any measurement $M$ containing that result.

Now consider that $\varpi(|\phi\rangle|\mu)$ is the probability of
obtaining some $\lambda\in\Lambda_{|\phi\rangle}$ when sampling $\mu$.
If $\mu\in\Delta_{|\psi\rangle}$ for some $|\psi\rangle\in\mathcal{P}(\mathcal{H})$,
then 
\begin{eqnarray}
\varpi(|\phi\rangle\,|\,\mu) & \eqdef & \int_{\Lambda_{|\phi\rangle}}\mathrm{d}\lambda\,\mu(\lambda)\\
 & = & \int_{\Lambda_{|\phi\rangle}}\mathrm{d}\lambda\,\mu(\lambda)\mathbb{P}_{M}(|\phi\rangle\,|\,\lambda)\\
 & \leq & \int_{\Lambda}\mathrm{d}\lambda\,\mu(\lambda)\mathbb{P}_{M}(|\phi\rangle\,|\,\lambda)=|\langle\phi|\psi\rangle|^{2}
\end{eqnarray}
by Eqs.~(\ref{eq:exact-quantum-statistics},\ref{eq:asymmetric-overlap}),
thus proving Eq.~(\ref{eq:asymmetric-trivial-bound}).\qed

\subsection{Anti-distinguishability and multi-partite asymmetric overlaps}

The main text states that if $\{|\psi\rangle,|\phi\rangle,|0\rangle\}$
is an anti-distinguishable triple, then $\Lambda_{|\psi\rangle}\cap\Lambda_{|\phi\rangle}\cap\Lambda_{|0\rangle}=\emptyset$
which further implies that $\varpi(|0\rangle,|\phi\rangle|\mu)=\varpi(|0\rangle|\mu)+\varpi(|\phi\rangle|\mu)$
$\forall\mu\in\Delta_{|\psi\rangle}$. Here, a more general statement,
necessary for the proofs of Thms.~\ref{thm:large-d-limit-theorem},\ref{thm:noisy-theorem},
is proved. Define the set $\mathcal{A}=\{|0\rangle,|\phi\rangle,...\}$
and let $\mu\in\Delta_{|\psi\rangle}$ be a preparation distribution
for a state $|\psi\rangle\not\in\mathcal{A}$. The statement is that
if each triple $\{|0\rangle,|\psi\rangle,|\phi\rangle\}$, where $|0\rangle,|\phi\rangle$
are unequal states from $\mathcal{A}$, is anti-distinguishable, then
Eq.~(\ref{eq:multipartite-anti-distinguishable-bound}) holds with
equality. 

Recall that $\varpi(|0\rangle,|\phi\rangle,...|\mu)$ is the probability
of obtaining a $\lambda\in\cup_{|a\rangle\in\mathcal{A}}\Lambda_{|a\rangle}$
by sampling from $\mu$, while, for every $|\phi\rangle\in\mathcal{A}$,
$\varpi(|\phi\rangle|\mu)$ is the probability of obtaining a $\lambda\in\Lambda_{|\phi\rangle}$
from $\mu$. The event corresponding to the probability $\varpi(|0\rangle,|\phi\rangle,...|\mu)$
must therefore be the disjunction of the events corresponding to each
probability $\varpi(|\phi\rangle\in\mathcal{A}|\mu)$. Applying Boole's
inequality therefore gives Eq.~(\ref{eq:multipartite-anti-distinguishable-bound}). 

Now suppose that each triple $\{|0\rangle,|\psi\rangle,|\phi\rangle\}$,
where $|0\rangle,|\phi\rangle$ are unequal states from $\mathcal{A}$,
is anti-distinguishable. Can the events corresponding to $\varpi(|0\rangle|\mu)$
and $\varpi(|\phi\rangle|\mu)$ occur simultaneously (or are they
mutually exclusive)? This is only possible if there exists a finite-measure
set of ontic states $\lambda\in\Lambda_{\mu}\cap\Lambda_{|0\rangle}\cap\Lambda_{|\phi\rangle}$.
It shall now be shown that anti-distinguishability and Eq.~(\ref{eq:exact-quantum-statistics})
prevent this.

Let $\chi\in\Delta_{|0\rangle},\nu\in\Delta_{|\phi\rangle}$ be any
relevant pair of preparation distributions and let $M=\{E_{\text{\textlnot}0},E_{\text{\textlnot}\psi},E_{\text{\textlnot}\phi}\}$
be the anti-distinguishing measurement for $\{|0\rangle,|\psi\rangle,|\phi\rangle\}$.
Equations~(\ref{eq:exact-quantum-statistics},\ref{eq:anti-distinguishing-measurement})
imply that
\begin{equation}
\int_{\Lambda_{\chi}}\mathrm{d}\lambda\,\chi(\lambda)\mathbb{P}_{M}(E_{\text{\textlnot}0}\,|\,\lambda)=\int_{\Lambda_{\mu}}\mathrm{d}\lambda\,\mu(\lambda)\mathbb{P}_{M}(E_{\text{\textlnot}\psi}\,|\,\lambda)=\int_{\Lambda_{\nu}}\mathrm{d}\lambda\,\nu(\lambda)\mathbb{P}_{M}(E_{\text{\textlnot}\phi}\,|\,\lambda)=0.
\end{equation}
These respectively imply the following: for almost all $\lambda\in\Lambda_{\chi}$,
$\mathbb{P}_{M}(E_{\text{\textlnot}0}|\lambda)=0$; for almost all
$\lambda\in\Lambda_{\mu}$, $\mathbb{P}_{M}(E_{\text{\textlnot}\psi}|\lambda)=0$;
and for almost all $\lambda\in\Lambda_{\nu}$, $\mathbb{P}_{M}(E_{\text{\textlnot}\phi}|\lambda)=0$.
Since this holds for arbitrary $\chi$ and $\nu$, it follows that\footnote{Similarly to the previous footnote, a fully rigorous treatment would
include a proof of this step, which is omitted for conceptual clarity
and since mathematical rigour has already been sacrificed for conceptual
clarity earlier in the paper.} for almost all $\lambda\in\Lambda_{|0\rangle}$, $\mathbb{P}_{M}(E_{\text{\textlnot}0}|\lambda)=0$,
and for almost all $\lambda\in\Lambda_{|\phi\rangle}$, $\mathbb{P}_{M}(E_{\text{\textlnot}\phi}|\lambda)=0$.

Therefore, for almost all $\lambda\in\Lambda_{\mu}\cap\Lambda_{|0\rangle}\cap\Lambda_{|\phi\rangle}$
it follows that $\mathbb{P}_{M}(E_{\text{\textlnot}0}|\lambda)=\mathbb{P}_{M}(E_{\text{\textlnot}\psi}|\lambda)=\mathbb{P}_{M}(E_{\text{\textlnot}\phi}|\lambda)=0$.
However, this is impossible since some outcome must occur in any measurement,
requiring $\mathbb{P}_{M}(E_{\text{\textlnot}0}|\lambda)+\mathbb{P}_{M}(E_{\text{\textlnot}\psi}|\lambda)+\mathbb{P}(E_{\text{\textlnot}\phi}|\lambda)=1$.
So $\Lambda_{\mu}\cap\Lambda_{|0\rangle}\cap\Lambda_{|\phi\rangle}$
must be of measure zero and the events corresponding to $\varpi(|0\rangle|\mu)$
and $\varpi(|\phi\rangle|\mu)$ cannot occur simultaneously---they
are mutually exclusive.

Since Boole's inequality holds with equality for mutually exclusive
events, it follows that Eq.~(\ref{eq:multipartite-anti-distinguishable-bound})
holds with equality whenever every such triple $\{|0\rangle,|\psi\rangle,|\phi\rangle\}$
is anti-distinguishable.\qed

\subsection{Unitary transformations never decrease ontic overlaps}

Consider quantum states $|\psi\rangle,|\phi\rangle\in\mathcal{P}(\mathcal{H})$,
a unitary transformation $\gamma\in\Gamma_{U}$, and preparation distribution
$\nu\in\Delta_{|\phi\rangle}$ so that under $\gamma$, $|\psi\rangle$
transforms to $U|\psi\rangle$ and $\nu$ to $\nu^{\prime}\in\Delta_{U|\phi\rangle}$.
By Eqs.~(\ref{eq:ontological-transformation},\ref{eq:asymmetric-overlap})
one finds
\begin{eqnarray}
\varpi(U|\psi\rangle\,|\,\nu^{\prime}) & = & \int_{\Lambda_{U|\psi\rangle}}\mathrm{d}\lambda^{\prime}\,\nu^{\prime}(\lambda^{\prime})\\
 & = & \int_{\Lambda_{U|\psi\rangle}}\mathrm{d}\lambda^{\prime}\int_{\Lambda}\mathrm{d}\lambda\,\nu(\lambda)\gamma(\lambda^{\prime}|\lambda)\\
 & \geq & \int_{\Lambda_{|\psi\rangle}}\mathrm{d}\lambda\,\nu(\lambda)\int_{\Lambda_{U|\psi\rangle}}\mathrm{d}\lambda^{\prime}\gamma(\lambda^{\prime}|\lambda).
\end{eqnarray}
Consider the transition probability $\int_{\Lambda_{U|\psi\rangle}}\mathrm{d}\lambda^{\prime}\gamma(\lambda^{\prime}|\lambda)$,
where $\lambda\in\Lambda_{|\psi\rangle}$. Suppose towards a contradiction
that this probability is less than unity $\int_{\Lambda_{U|\psi\rangle}}\mathrm{d}\lambda^{\prime}\gamma(\lambda^{\prime}|\lambda)<1$
for some finite measure of $\lambda\in\Lambda_{|\psi\rangle}$. This
implies that\footnote{Once again, such a more rigorous formulation of the problem would
require a full justification of this step.} there is some preparation $\mu\in\Delta_{|\psi\rangle}$ of $|\psi\rangle$
such that 
\begin{eqnarray}
1 & > & \int_{\Lambda}\mathrm{d}\lambda\,\mu(\lambda)\int_{\Lambda_{U|\psi\rangle}}\mathrm{d}\lambda^{\prime}\gamma(\lambda^{\prime}|\lambda)\\
 & = & \int_{\Lambda_{U|\psi\rangle}}\mathrm{d}\lambda^{\prime}\mu^{\prime}(\lambda^{\prime})\\
 & = & \int_{\Lambda}\mathrm{d}\lambda^{\prime}\mu^{\prime}(\lambda^{\prime})
\end{eqnarray}
by Eq.~(\ref{eq:ontological-transformation}) where $\mu^{\prime}\in\Delta_{U|\psi\rangle}$
is obtained from $\mu$ via $\gamma$, which is a contradiction since
preparations must always produce some ontic state $\int_{\Lambda}{\rm d}\lambda^{\prime}\mu^{\prime}(\lambda^{\prime})=1$.
Therefore, $\int_{\Lambda_{U|\psi\rangle}}\mathrm{d}\lambda^{\prime}\gamma(\lambda^{\prime}|\lambda)=1$
and so 
\begin{equation}
\varpi(U|\psi\rangle\,|\,\nu^{\prime})\geq\int_{\Lambda_{|\psi\rangle}}{\rm d}\lambda\,\nu(\lambda)=\varpi(|\psi\rangle\,|\,\nu),
\end{equation}
thus proving Eq.~(\ref{eq:first-thm-monotonicity}).

The same result also holds for the symmetric overlap {[}Eq.~(\ref{eq:quantum-state-symmetric}){]}
between any pair of quantum states $|\psi\rangle,|\phi\rangle$. Consider
any pair $\mu\in\Delta_{|\psi\rangle},\nu\in\Delta_{|\phi\rangle}$,
then $\omega(\mu,\nu)$ is simply twice the optimal average probability
of error when attempting to guess which of $\mu$ or $\nu$ a given
$\lambda\in\Lambda$ was sampled from. For any stochastic map $\gamma$
that transforms $\mu$ to $\mu^{\prime}$ and $\nu$ to $\nu^{\prime}$,
a strategy for distinguishing $\mu^{\prime},\nu^{\prime}$ is also
a strategy for distinguishing $\mu,\nu$. Therefore, the optimal strategy
for distinguishing $\mu^{\prime},\nu^{\prime}$ cannot, by definition,
have a lower probability of error than the optimal strategy for distinguishing
$\mu,\nu$. It immediately follows that 
\begin{equation}
\omega(\mu^{\prime},\nu^{\prime})\geq\omega(\mu,\nu)\label{eq:symmetric-monotonicity}
\end{equation}
and by Eq.~(\ref{eq:symmetric-state-trivial-bound}) that $\omega(U|\psi\rangle,U|\phi\rangle)\geq\omega(|\psi\rangle,|\phi\rangle)$
for any unitary $U$.\qed

\subsection{Theorem \ref{thm:large-d-limit-theorem}: Bounding general state
overlaps}

The proof strategy is almost identical to that of Thm.~\ref{thm:superpositions-are-real},
but modified to make use of higher dimensional systems.

Any such $|\psi\rangle$ can be written in the form of Eq.~(\ref{eq:Psi-construction})
for some ONB $\mathcal{B}^{\prime}=\{|0\rangle\}\cup\{|i^{\prime}\rangle\}_{i=1}^{d-1}$
and where $\beta\eqdef\sqrt{2}\alpha^{\frac{3}{2}}$. In this case
$|\alpha|^{2}+|\beta|^{2}=|\alpha|^{2}+2|\alpha|^{3}<1$ so the construction
remains possible. Similarly, a set of states $\{|\phi_{i}\rangle\}_{i=3}^{d-1}$
can be defined with respect to the same basis by 
\begin{equation}
|\phi_{i}\rangle\eqdef\delta|0\rangle+\eta|1^{\prime}\rangle+\kappa|i^{\prime}\rangle
\end{equation}
with $\delta\eqdef1-2\alpha^{2}$ and $\eta\eqdef\sqrt{2}\alpha^{\frac{3}{2}}$.
Again, this is possible since $|\delta|^{2}+|\eta|^{2}=(1-2\alpha^{2})^{2}+2\alpha^{3}<1$.
Note that the definitions of $\beta$ and $\eta$ have changed from
those used in Thm.~\ref{thm:superpositions-are-real}.

It may be verified by Eq.~(\ref{eq:antidistinguishing}) that both
$\{|0\rangle,|\psi\rangle,|\phi_{i}\rangle\}$ and $\{|\psi\rangle,|\phi_{i}\rangle,|\phi_{j}\rangle\}$
are anti-distinguishable triples for all $i\neq j$.

Note that $|\langle\phi_{i}|\psi\rangle|^{2}=\alpha^{2}=|\langle0|\psi\rangle|^{2}$
for all $i$, so there exist stabiliser unitaries $\{U_{i}\}_{i=3}^{d-1}\subset\mathcal{S}_{|\psi\rangle}$
for which $U_{i}|0\rangle=|\phi_{i}\rangle$. Consider preparing $|\psi\rangle$
via some $\mu\in\Delta_{|\psi\rangle}$ then transforming with $U_{i}$
via any $\gamma_{i}\in\Gamma_{U_{i}}$. By assumption, preparations
of $|\psi\rangle$ are non-contextual with respect to such stabiliser
unitaries, so $\mu$ simply transforms to itself. Therefore, by Eq.~(\ref{eq:first-thm-monotonicity}),
it is found that 
\begin{equation}
\varpi(|\phi_{i}\rangle\,|\,\mu)\geq\varpi(|0\rangle\,|\,\mu)\quad\forall i.
\end{equation}

So, prepare the state $|\psi\rangle$ via $\mu$ and then perform
a measurement $M$ in the $\mathcal{B}^{\prime}$ basis. Since $|0\rangle$
and $|1^{\prime}\rangle$ are the only measurement outcomes compatible
with $\lambda\in\Lambda_{|\psi\rangle}\cap(\Lambda_{|0\rangle}\cup_{i=3}^{d-1}\Lambda_{|\phi_{i}\rangle})$,
then the asymmetric overlap with these states must lower bound the
probability of obtaining either $|0\rangle$ or $|1^{\prime}\rangle$.
One therefore finds that 
\begin{eqnarray}
\mathbb{P}_{M}(|0\rangle\vee|1^{\prime}\rangle\,|\,\mu) & \geq & \varpi(|0\rangle,|\phi_{3}\rangle,...,|\phi_{d-1}\rangle\,|\,\mu)\nonumber \\
 & = & \varpi(|0\rangle\,|\,\mu)+\sum_{i=3}^{d-1}\varpi(|\phi_{i}\rangle\,|\,\mu)\nonumber \\
 & \geq & (d-2)\varpi(|0\rangle\,|\,\mu)
\end{eqnarray}
where the second line follows because each of the sets $\{|0\rangle,|\psi\rangle,|\phi_{i}\rangle\}$
and $\{|\psi\rangle,|\phi_{i}\rangle,|\phi_{j}\rangle\}$ are anti-distinguishable.
So if quantum predictions are exactly reproduced, one finds that $\mathbb{P}_{M}(|0\rangle\vee|1^{\prime}\rangle|\mu)=\alpha^{2}+2\alpha^{3}$
and 
\begin{equation}
\varpi(|0\rangle\,|\,\mu)\leq\alpha^{2}\left(\frac{1+2\alpha}{d-2}\right).
\end{equation}
Which completes the proof.\qed

\subsection{Theorem \ref{thm:noisy-theorem}: A noise-tolerant bound on the symmetric
overlap}

This proof uses the the assumptions, notation, and constructions from
Thm.~\ref{thm:large-d-limit-theorem}, except this time it is only
assumed that the ontological model reproduces quantum probabilities
to within some additive error $\epsilon\in(0,1]$. It will also be
necessary to define the tri-partite symmetric overlap between three
probability distributions $\mu,\nu,\chi$ \citep{Leifer14b}
\begin{equation}
\omega(\mu,\nu,\chi)\eqdef\int_{\Lambda}\mathrm{d}\lambda\,\min\{\mu(\lambda),\nu(\lambda),\chi(\lambda)\}.\label{eq:tripartite-symmetric-overlap}
\end{equation}

Consider any pair of preparation distributions $\mu\in\Delta_{|\psi\rangle}$,
$\nu\in\Delta_{|0\rangle}$. From Thm.~\ref{thm:large-d-limit-theorem}
it is known that there exist $U_{i}\in\mathcal{S}_{|\psi\rangle}$
such that $U_{i}|0\rangle=|\phi_{i}\rangle$. For each $U_{i}$ consider
any corresponding stochastic map $\gamma_{i}\in\Gamma_{U_{i}}$. By
assumption, preparations of $|\psi\rangle$ are non-contextual with
respect to stabiliser unitaries so each $\gamma_{i}$ maps $\mu$
to itself. Let each $\gamma_{i}$ map $\nu$ to some $\chi_{i}\in\Delta_{|\phi_{i}\rangle}$.
For notational convenience, let $|\phi_{0}\rangle\eqdef|0\rangle$
and $\chi_{0}\eqdef\nu$ then define the sets $\tilde{I}\eqdef\{3,...,d-1\}$
and $I\eqdef\{0\}\cup\tilde{I}$. By Eq.~(\ref{eq:symmetric-monotonicity})
it is therefore seen that
\begin{equation}
\omega(\mu,\chi_{i})\geq\omega(\mu,\nu),\quad\forall i\in\tilde{I}.\label{eq:noisy-prep-noncontextual}
\end{equation}

Consider a preparation of $|\psi\rangle$ via $\mu$, followed by
a measurement $M$ in the basis $\mathcal{B}^{\prime}$. Similarly
to Thm.~\ref{thm:large-d-limit-theorem}, the aim is to bound $\omega(\mu,\nu)$
by considering the probability of obtaining either of the measurement
outcomes $|0\rangle$ or $|1^{\prime}\rangle$, given by
\begin{equation}
\mathbb{P}_{M}(|0\rangle\vee|1^{\prime}\rangle\,|\,\mu)\leq\alpha^{2}+\beta^{2}+\epsilon.\label{eq:thm4prob}
\end{equation}
The trick is to do this in such a way that all possible errors are
accounted for.

In order to link this quantum probability to symmetric overlaps, consider
the following subsets of $\Lambda$.
\begin{itemize}
\item For each $i\in I$ consider $\Omega_{i}\eqdef\{\lambda\in\Lambda\,:\,0<\mu(\lambda)\leq\chi_{i}(\lambda)\}$.
Roughly, $\Omega_{i}$ is the region of the overlap between $\mu$
and $\chi_{i}$ for which $\mu$ is smaller than $\chi_{i}$.
\item For each $i\in I$ consider $\Theta_{i}\eqdef\{\lambda\in\Lambda\,:\,0<\chi_{i}(\lambda)<\mu(\lambda);\,\forall j<i,\,\chi_{j}(\lambda)\leq\chi_{i}(\lambda);\,\forall j>i,\,0<\chi_{j}(\lambda)<\chi_{i}(\lambda)\}$.
Roughly, this is the region of the overlap between $\mu$ and $\chi_{i}$
for which $\chi_{i}$ is greater than all other $\chi_{j\neq i}$,
but smaller than $\mu$.
\item For each $i<j\in I$ consider $\Theta_{i}^{j}\eqdef\{\lambda\in\Lambda\,:\,0<\chi_{i}(\lambda)\leq\chi_{j}(\lambda);\,\chi_{i}(\lambda)<\mu(\lambda)\}$.
Roughly, this is the region of the tri-partite overlap of $\mu,\chi_{i},\chi_{j}$
in which $\chi_{i}$ is the minimum of the three.
\item Similarly, for each $i>j\in I$ consider $\Theta_{i}^{j}\eqdef\{\lambda\in\Lambda\,:\,0<\chi_{i}(\lambda)<\chi_{j}(\lambda);\,\chi_{i}(\lambda)<\mu(\lambda)\}$.
\item For every unequal pair $i,j\in I$, let $\Omega_{ij}=\Omega_{i}\cap\Omega_{j}$.
\end{itemize}
Note that these sets are defined to be disjoint, for $i\neq j$: $\Theta_{i}\cap\Omega_{j}=\Theta_{i}\cap\Theta_{j}=\Theta_{i}\cap\Theta_{i}^{j}=\emptyset$.

The point of these subsets is the way in which they relate to symmetric
overlaps. From the definitions of symmetric overlaps it is not difficult
to verify that
\begin{eqnarray}
\omega(\mu,\chi_{i}) & = & \int_{\Omega_{i}}{\rm d}\lambda\,\mu(\lambda)+\int_{\Theta_{i}\cup\left[\cup_{j\neq i}\Theta_{i}^{j}\right]}\mathrm{d}\lambda\,\chi_{i}(\lambda)\label{eq:noisy-thm-bipartite-symmetric}\\
\omega(\mu,\chi_{i},\chi_{j\neq i}) & = & \int_{\Omega_{ij}}\mathrm{d}\lambda\,\mu(\lambda)+\int_{\Theta_{i}^{j}}\mathrm{d}\lambda\,\chi_{i}(\lambda)+\int_{\Theta_{j}^{i}}\mathrm{d}\lambda\,\chi_{j}(\lambda).\label{eq:noisy-thm-tripartite-symmetric}
\end{eqnarray}

Proceed by separating the probability Eq.~(\ref{eq:thm4prob}) according
to subsets in which $\lambda$ may obtain: 
\begin{eqnarray}
\mathbb{P}_{M}(|0\rangle\vee|1^{\prime}\rangle\,|\,\mu) & \geq & \sum_{i\in I}\mathbb{P}_{M}(|0\rangle\vee|1^{\prime}\rangle,\lambda\in\Omega_{i}\,|\,\mu)\nonumber \\
 &  & +\sum_{i\in I}\mathbb{P}_{M}(|0\rangle\vee|1^{\prime}\rangle,\lambda\in\Theta_{i}\,|\,\mu)\nonumber \\
 &  & -\sum_{i,j<i}\mathbb{P}_{M}(|0\rangle\vee|1^{\prime}\rangle,\lambda\in\Omega_{ij}\,|\,\mu),\\
 & \geq & \sum_{i\in I}\mathbb{P}_{M}(|0\rangle\vee|1^{\prime}\rangle,\lambda\in\Omega_{i}\,|\,\mu)\nonumber \\
 &  & +\sum_{i\in I}\mathbb{P}_{M}(|0\rangle\vee|1^{\prime}\rangle,\lambda\in\Theta_{i}\,|\,\mu)\nonumber \\
 &  & -\sum_{i,j<i}\int_{\Omega_{ij}}\mathrm{d}\lambda\mu(\lambda).\label{eq:noisy-thm-analysis}
\end{eqnarray}
The final line follows simply because $\mathbb{P}_{M}(|0\rangle\vee|1^{\prime}\rangle,\lambda\in\Omega_{ij}|\mu)\leq\mathbb{P}_{M}(\lambda\in\Omega_{ij}|\mu)=\int_{\Omega_{ij}}\mathrm{d}\lambda\mu(\lambda)$.

For the $i=0$ term in the first line of Eq.~(\ref{eq:noisy-thm-analysis}),
define the function $\xi(\lambda)\eqdef1-\mathbb{P}_{M}(|0\rangle|\lambda)$
so that 
\begin{eqnarray}
\mathbb{P}_{M}(|0\rangle\vee|1^{\prime}\rangle,\lambda\in\Omega_{0}\,|\,\mu) & = & \int_{\Omega_{0}}\mathrm{d}\lambda\,\mu(\lambda)\left\{ \mathbb{P}_{M}(|0\rangle\,|\,\lambda)+\mathbb{P}_{M}(|1^{\prime}\rangle\,|\,\lambda)\right\} \\
 & \geq & \int_{\Omega_{0}}\mathrm{d}\lambda\,\mu(\lambda)\mathbb{P}_{M}(|0\rangle\,|\,\lambda)\\
 & = & \int_{\Omega_{0}}\mathrm{d}\lambda\,\mu(\lambda)-\int_{\Omega_{0}}\mathrm{d}\lambda\,\mu(\lambda)\xi(\lambda)\\
 & \geq & \int_{\Omega_{0}}\mathrm{d}\lambda\,\mu(\lambda)-\int_{\Omega_{0}}\mathrm{d}\lambda\,\nu(\lambda)\xi(\lambda).
\end{eqnarray}
This can be simplified by noting that, for any $\Omega\subseteq\Lambda$,
\begin{eqnarray}
\int_{\Omega}\mathrm{d}\nu(\lambda)\xi(\lambda) & = & \int_{\Omega}\mathrm{d}\lambda\,\nu(\lambda)-\int_{\Lambda}\mathrm{d}\lambda\,\nu(\lambda)\mathbb{P}_{M}(|0\rangle\,|\,\lambda)+\int_{\Lambda\backslash\Omega}\mathrm{d}\lambda\,\nu(\lambda)\mathbb{P}_{M}(|0\rangle\,|\,\lambda)\\
 & \leq & \int_{\Omega}\mathrm{d}\lambda\,\nu(\lambda)-1+\epsilon+\int_{\Lambda\backslash\Omega}\mathrm{d}\lambda\,\nu(\lambda)\\
 & = & \epsilon
\end{eqnarray}
so that
\begin{equation}
\mathbb{P}_{M}(|0\rangle\vee|1^{\prime}\rangle,\lambda\in\Omega_{0}\,|\,\mu)\geq\int_{\Omega_{0}}\mathrm{d}\lambda\,\mu(\lambda)-\epsilon.\label{eq:first-terms-i=00003D0}
\end{equation}

The $i\in\tilde{I}$ terms of the first line of Eq.~(\ref{eq:noisy-thm-analysis})
follow in a similar way. Define $\zeta_{i}(\lambda)\eqdef1-\mathbb{P}_{M}(|0\rangle|\lambda)-\mathbb{P}_{M}(|1^{\prime}\rangle|\lambda)-\mathbb{P}_{M}(|i^{\prime}\rangle|\lambda)$
so that 
\begin{eqnarray}
\mathbb{P}_{M}(|0\rangle\vee|1^{\prime}\rangle,\lambda\in\Omega_{i}\,|\,\mu) & = & \int_{\Omega_{i}}\mathrm{d}\lambda\,\mu(\lambda)\left\{ \mathbb{P}_{M}(|0\rangle\,|\,\lambda)+\mathbb{P}_{M}(|1^{\prime}\rangle\,|\,\lambda)\right\} \\
 & = & \int_{\Omega_{i}}\mathrm{d}\lambda\,\mu(\lambda)-\int_{\Omega_{i}}\mathrm{d}\lambda\,\mu(\lambda)\left\{ \zeta_{i}(\lambda)+\mathbb{P}_{M}(|i^{\prime}\rangle\,|\,\lambda)\right\} \\
 & \geq & \int_{\Omega_{i}}\mathrm{d}\lambda\,\mu(\lambda)-\int_{\Omega_{i}}\mathrm{d}\lambda\,\chi_{i}(\lambda)\zeta_{i}(\lambda)-\epsilon\\
 & \geq & \int_{\Omega_{i}}\mathrm{d}\lambda\,\mu(\lambda)-2\epsilon.\label{eq:first-terms-iinI}
\end{eqnarray}
Together, Eqs.~(\ref{eq:noisy-thm-analysis},\ref{eq:first-terms-i=00003D0},\ref{eq:first-terms-iinI})
produce
\begin{eqnarray}
\mathbb{P}_{M}(|0\rangle\vee|1^{\prime}\rangle\,|\,\mu) & \geq & \sum_{i\in I}\int_{\Omega_{i}}\mathrm{d}\lambda\,\mu(\lambda)-\sum_{i,j<i}\int_{\Omega_{ij}}\mathrm{d}\lambda\,\mu(\lambda)-(2d-5)\epsilon\nonumber \\
 &  & +\sum_{i\in I}\mathbb{P}_{M}(|0\rangle\vee|1^{\prime}\rangle,\lambda\in\Theta_{i}\,|\,\mu).\label{eq:noisy-thm-analysis-2}
\end{eqnarray}

The $i=0$ term of the second line of Eq.~(\ref{eq:noisy-thm-analysis-2})
can be bounded as follows
\begin{eqnarray}
\mathbb{P}_{M}(|0\rangle\vee|1^{\prime}\rangle,\lambda\in\Theta_{0}\,|\,\mu) & \geq & \int_{\Theta_{0}}\mathrm{d}\lambda\,\mu(\lambda)\mathbb{P}_{M}(|0\rangle\,|\,\lambda)\\
 & \geq & \int_{\Theta_{0}}\mathrm{d}\lambda\,\nu(\lambda)\mathbb{P}_{M}(|0\rangle\,|\,\lambda)\\
 & = & \int_{\Lambda}\mathrm{d}\lambda\,\nu(\lambda)\mathbb{P}_{M}(|0\rangle\,|\,\lambda)-\int_{\Lambda\backslash\Theta_{0}}\mathrm{d}\lambda\,\nu(\lambda)\mathbb{P}_{M}(|0\rangle\,|\,\lambda)\\
 & \geq & 1-\epsilon-\int_{\Lambda\backslash\Theta_{0}}\mathrm{d}\lambda\,\nu(\lambda)=\int_{\Theta_{0}}\mathrm{d}\lambda\,\nu(\lambda)-\epsilon.\label{eq:second-terms-i=00003D0}
\end{eqnarray}
The $i\in\tilde{I}$ terms of the second line of Eq.~(\ref{eq:noisy-thm-analysis-2})
can be similarly bounded
\begin{eqnarray}
\mathbb{P}_{M}(|0\rangle\vee|1^{\prime}\rangle,\lambda\in\Theta_{i}\,|\,\mu) & = & \int_{\Theta_{i}}\mathrm{d}\lambda\,\mu(\lambda)\left\{ \mathbb{P}_{M}(|0\rangle\,|\,\lambda)+\mathbb{P}_{M}(|1^{\prime}\rangle\,|\,\lambda)+\mathbb{P}_{M}(|i^{\prime}\rangle\,|\,\lambda)\right\} -\int_{\Theta_{i}}\mathrm{d}\lambda\,\mu(\lambda)\mathbb{P}_{M}(|i^{\prime}\rangle\,|\,\lambda)\\
 & \geq & \int_{\Theta_{i}}\mathrm{d}\lambda\,\chi_{i}(\lambda)\left\{ \mathbb{P}_{M}(|0\rangle\,|\,\lambda)+\mathbb{P}_{M}(|1^{\prime}\rangle\,|\,\lambda)+\mathbb{P}_{M}(|i^{\prime}\rangle\,|\,\lambda)\right\} -\epsilon\\
 & \geq & (1-\epsilon)-\int_{\Lambda\backslash\Theta_{i}}\mathrm{d}\lambda\,\chi_{i}(\lambda)-\epsilon=\int_{\Theta_{i}}\mathrm{d}\lambda\,\chi_{i}(\lambda)-2\epsilon.\label{eq:second-terms-iinI}
\end{eqnarray}
So now combining Eqs.~(\ref{eq:noisy-thm-analysis-2},\ref{eq:second-terms-i=00003D0},\ref{eq:second-terms-iinI})
it is found that
\begin{eqnarray}
\mathbb{P}_{M}(|0\rangle\vee|1^{\prime}\rangle\,|\,\mu) & \geq & \sum_{i\in I}\left\{ \int_{\Omega_{i}}\mathrm{d}\lambda\,\mu(\lambda)+\int_{\Theta_{i}}\mathrm{d}\lambda\,\chi_{i}(\lambda)\right\} \nonumber \\
 &  & -\sum_{i,j<i}\int_{\Omega_{ij}}\mathrm{d}\lambda\mu(\lambda)-2(2d-5)\epsilon.\label{eq:noisy-thm-analysis-3}
\end{eqnarray}

Equation~(\ref{eq:noisy-thm-analysis-3}) can be further reduced
by adding any negative quantity. For example, consider Boole's inequality
\begin{equation}
\sum_{i\in I}\int_{\cup_{j\neq i}\Theta_{i}^{j}}\mathrm{d}\lambda\,\chi_{i}(\lambda)-\sum_{i,j\neq i}\int_{\Theta_{i}^{j}}\mathrm{d}\lambda\,\chi_{i}(\lambda)\leq0.
\end{equation}
Therefore, Eq.~(\ref{eq:noisy-thm-analysis-3}) reduces to
\begin{eqnarray}
\mathbb{P}_{M}(|0\rangle\vee|1^{\prime}\rangle\,|\,\mu) & \geq & \sum_{i\in I}\left\{ \int_{\Omega_{i}}\mathrm{d}\lambda\,\mu(\lambda)+\int_{\Theta_{i}\cup\left[\cup_{j\neq i}\Theta_{i}^{j}\right]}\mathrm{d}\lambda\,\chi_{i}\left(\lambda\right)\right\} \nonumber \\
 &  & -\sum_{i,j<i}\int_{\Omega_{ij}}\mathrm{d}\lambda\mu(\lambda)-\sum_{i,j\neq i}\int_{\Theta_{i}^{j}}\mathrm{d}\lambda\,\chi_{i}(\lambda)\nonumber \\
 &  & -2(2d-5)\epsilon.\label{eq:noisy-thm-analysis-4}
\end{eqnarray}
This can be further simplified by noting
\begin{eqnarray}
\sum_{i,j\neq i}\int_{\Theta_{i}^{j}}\mathrm{d}\lambda\,\chi_{i}(\lambda) & = & \sum_{i,j<i}\int_{\Theta_{i}^{j}}\mathrm{d}\lambda\,\chi_{i}(\lambda)+\sum_{i,j>i}\int_{\Theta_{i}^{j}}\mathrm{d}\lambda\,\chi_{i}(\lambda)\\
 & = & \sum_{i,j<i}\int_{\Theta_{i}^{j}}\mathrm{d}\lambda\,\chi_{i}(\lambda)+\sum_{j,i<j}\int_{\Theta_{i}^{j}}\mathrm{d}\lambda\,\chi_{i}(\lambda)\\
 & = & \sum_{i,j<i}\left\{ \int_{\Theta_{i}^{j}}\mathrm{d}\lambda\,\chi_{i}(\lambda)+\int_{\Theta_{j}^{i}}\mathrm{d}\lambda\,\chi_{j}(\lambda)\right\} ,
\end{eqnarray}
so that Eq.~(\ref{eq:noisy-thm-analysis-4}) becomes
\begin{equation}
\mathbb{P}_{M}(|0\rangle\vee|1^{\prime}\rangle\,|\,\mu)\geq\sum_{i\in I}\omega(\mu,\chi_{i})-\sum_{i,j<i}\omega(\mu,\chi_{i},\chi_{j})-2(2d-5)\epsilon\label{eq:noisy-thm-analysis-5}
\end{equation}
having used Eqs.(\ref{eq:noisy-thm-bipartite-symmetric},\ref{eq:noisy-thm-tripartite-symmetric}).

As a final step, consider how the tripartite symmetric overlaps are
bounded by $\epsilon$. Consider the measurement $M^{\prime}=\{E_{\text{\textlnot\ensuremath{\psi}}},E_{\text{\textlnot}i},E_{\text{\textlnot}j}\}$
which anti-distinguishes $\{|\psi\rangle,|\phi_{i}\rangle|\phi_{j}\rangle\}$,
so that
\begin{eqnarray}
\int_{\Lambda}\mathrm{d}\mu(\lambda)\mathbb{P}_{M^{\prime}}(E_{\text{\textlnot}\psi}\,|\,\lambda) & \leq & \epsilon,\\
\int_{\Lambda}\mathrm{d}\chi_{i}(\lambda)\mathbb{P}_{M^{\prime}}(E_{\text{\textlnot}i}\,|\,\lambda) & \leq & \epsilon,\\
\int_{\Lambda}\mathrm{d}\chi_{j}(\lambda)\mathbb{P}_{M^{\prime}}(E_{\text{\textlnot}j}\,|\,\lambda) & \leq & \epsilon.
\end{eqnarray}
Conservation of probability requires that 
\begin{equation}
\mathbb{P}_{M^{\prime}}(E_{\text{\textlnot}\psi}\,|\,\lambda)+\mathbb{P}_{M^{\prime}}(E_{\text{\textlnot}i}\,|\,\lambda)+\mathbb{P}_{M^{\prime}}(E_{\text{\textlnot}j}\,|\,\lambda)=1
\end{equation}
for all $\lambda\in\Lambda$. Consider the tripartite symmetric overlap,
Eqs.~(\ref{eq:tripartite-symmetric-overlap},\ref{eq:noisy-thm-tripartite-symmetric}),
for $\mu,\chi_{i},\chi_{j}$. Then
\begin{eqnarray}
\omega(\mu,\chi_{i},\chi_{j}) & = & \int_{\Lambda}\mathrm{d}\lambda\,\min\{\mu(\lambda),\chi_{i}(\lambda),\chi_{j}(\lambda)\}\left\{ \mathbb{P}_{M^{\prime}}(E_{\text{\textlnot}\psi}\,|\,\lambda)+\mathbb{P}_{M^{\prime}}(E_{\text{\textlnot}i}\,|\,\lambda)+\mathbb{P}_{M^{\prime}}(E_{\text{\textlnot}j}\,|\,\lambda)\right\} \\
 & \leq & \int_{\Lambda}\mathrm{d}\mu(\lambda)\mathbb{P}_{M^{\prime}}(E_{\text{\textlnot}\psi}\,|\,\lambda)+\int_{\Lambda}\mathrm{d}\chi_{i}(\lambda)\mathbb{P}_{M^{\prime}}(E_{\text{\textlnot}i}\,|\,\lambda)+\int_{\Lambda}\mathrm{d}\chi_{j}(\lambda)\mathbb{P}_{M^{\prime}}(E_{\text{\textlnot}j}\,|\,\lambda)\\
 & \leq & 3\epsilon.\label{eq:tripartite-symmetric-bound}
\end{eqnarray}

Applying Eq.~(\ref{eq:tripartite-symmetric-bound}) to Eq.~(\ref{eq:noisy-thm-analysis-5}),
one finds that
\begin{eqnarray}
\mathbb{P}_{M}(|0\rangle\vee|1^{\prime}\rangle\,|\,\mu) & \geq & \sum_{i\in I}\omega(\mu,\chi_{i})-\frac{3}{2}(d-3)(d-2)\epsilon-2(2d-5)\epsilon\\
 & \geq & (d-2)\omega(\mu,\nu)-\frac{1}{2}(3d^{2}-7d-2)\epsilon\label{eq:thm4-total-bound}
\end{eqnarray}
having used Eq.~(\ref{eq:noisy-prep-noncontextual}). Combining Eqs.~(\ref{eq:thm4prob},\ref{eq:thm4-total-bound})
one obtains an upper bound on $\omega(\mu,\nu)$ for any $\mu\in\Delta_{|\psi\rangle},\nu\in\Delta_{|0\rangle}$,
which must be greater than or equal to the least upper bound, Eq.~(\ref{eq:quantum-state-symmetric}),
finally giving 
\begin{equation}
\omega(|0\rangle,|\psi\rangle)\leq\alpha^{2}\left(\frac{1+2\alpha}{d-2}\right)+\frac{(3d^{2}-7d)}{2(d-2)}\epsilon.
\end{equation}
This completes the proof.\qed

A note is in order about the tightness of this bound, assuming arbitrarily
small $\epsilon$. At $d=4$, this bound cannot improve upon that
of Eq.~(\ref{eq:symmetric-state-trivial-bound}) for any $\alpha^{2}\in(0,\frac{1}{4})$.
At $d=5$, an improvement is possible for some values of $\alpha$.
It is only for $d>5$ that this bound is capable of improving upon
Eq.~(\ref{eq:symmetric-state-trivial-bound}) for all values of $\alpha\in(0,\frac{1}{4})$.
This is because the theorem extends the methods of Thms.~\ref{thm:superpositions-are-real},\ref{thm:large-d-limit-theorem},
which are closely linked to the asymmetric overlap, to the symmetric
overlap. 

Clearly the error model used here is very simplistic: it has been
assumed that some $\epsilon>0$ can be used to bound the deviation
of all probabilities from the quantum predictions. Another source
of possible error is in the use of stabiliser unitaries for $|\psi\rangle$.
To obtain Eq.~(\ref{eq:thm4-total-bound}) one uses Eq.~(\ref{eq:noisy-prep-noncontextual})
which requires that the $\chi_{i}$ are obtained from $\mu$ by a
transformation implementing a stabiliser unitary. Any experiment would
also have to engage in the problem of how to account for errors in
the implementation of the stabiliser unitary.

It may be possible to improve on the error term above by more carefully
using higher Bonferroni inequalities \citep{Rohatgi11}, rather than
just Boole's inequality as used above. Considering quad-partite and
higher-order overlaps (rather than stopping at the tripartite overlap,
as done here) may also improve the error. Doing so may improve the
scaling with $d$.

\section{\label{sec:Justifying-Preparation-Non-conte}Justifying Preparation
Non-contextuality with respect to Stabiliser Unitaries}

The ontological models formalism combines fundamental objective ontology
and operational notions. The fundamental ontology is reflected in
the idea that ontic states represent actual states of affairs, independently
of any other theories an observer might use to describe the same system.
On the other hand, the only way to reason about this largely-unspecified
ontological level is operationally: how does it respond to preparations,
transformations, and measurements that we can actually perform?

An assumption of non-contextuality is an assumption about these operational
bridges between our capabilities and the ontology. The following argument
is designed to defend the idea that operational preparations of some
pure quantum state $|\psi\rangle\in\mathcal{P}(\mathcal{H})$ can
reasonably be assumed to be non-contextual with respect to stabiliser
unitaries $\mathcal{S}_{|\psi\rangle}$ of the same state.

Any specific operational method for preparing some state $|\psi\rangle\in\mathcal{P}(\mathcal{H})$
may be thought of as a black box which the system is fed into. When
the system is fed out of the box, it is promised that the box has
prepared the system in state $|\psi\rangle$ according to some specific
method. In terms of ontological models, any preparation distribution
$\mu\in\Delta_{|\psi\rangle}$ can be considers in terms of such a
box.

Suppose you design some experiment, which involves preparing $|\psi\rangle$
via $\mu$. Scientists implementing that experiment would obtain the
corresponding black box to be sure that the method corresponding to
$\mu$ is indeed used. Once prepared, the system will the need to
be presented to other pieces of apparatus. However, there will always
be variation in how the system is treated between preparation and
the action of any other apparatus, any amount of motion or passage
of time or other (seemingly innocuous) treatment amounts to applying
some unitary $U$ to the system. Each scientist will, no doubt, be
careful to ensure that the system is not disturbed from its preparation
state, so it can be safely assumed that any such $U$ is a stabiliser
unitary $U\in\mathcal{S}_{|\psi\rangle}$. However, the point remains
that some unknown $U\in\mathcal{S}_{|\psi\rangle}$ is inevitably
applied to the system after preparation via $\mu$, and this can never
be perfectly accounted for.

Therefore, to analyse the result of the experiment, one has to allow
for some unknown $U\in\mathcal{S}_{|\psi\rangle}$ to by applied (via
some unknown $\gamma\in\Gamma_{U}$) after preparation of $|\psi\rangle$
via $\mu$. On this minimally realistic operational level, an arbitrary
preparation distribution $\mu$ can never be prepared unscathed, one
has to account for the inevitable, unknown, subsequent stabiliser
unitary. The effective preparation distribution that one must therefore
use to describe the experiment has to be one that is non-contextual
with respect to such transformations, allowing the experiment to still
be analysed despite the application of an unknown $U\in\mathcal{S}_{|\psi\rangle}$.

One must be careful to only consider operational features that are
not, even in principle, impossible to reliably perform. Since the
sets of preparation distributions for any given quantum states are,
in the end, operational in character, one may safely restrict to preparation
distributions that satisfy certain sensible realistic requirements.
The above heuristic argument aims to establish pure state preparation
non-contextuality with respect to stabiliser unitaries as such a realistic
requirement.
\end{document}